\documentclass[11pt]{article}
\usepackage{epsf}
\usepackage{amsmath}
\usepackage{epsfig}
\usepackage{times}
\usepackage{amssymb}
\usepackage{amsthm}
\usepackage{setspace}
\usepackage{cite}

\usepackage{algorithmic}  
\usepackage{algorithm}

\usepackage{shadow}
\usepackage{fancybox}
\usepackage{fancyhdr}

\def\w{{\bf w}}

\def\y{{\bf y}}

\def\x{{\bf x}}

\def\x{{\mathbf x}}

\def\w{{\bf w}}

\def\x{{\bf x}}
\def\y{{\bf y}}

\def\b{{\bf b}}

\def\h{{\bf h}}

\def\be{\begin{equation}}
\def\ee{\end{equation}}
\def\ba{\left[\begin{array}}
\def\ea{\end{array}\right]}

\def\w{{\bf w}}

\def\x{{\bf x}}
\def\y{{\bf y}}

\def\b{{\bf b}}

\def\1{{\bf 1}}

\def\0{{\bf 0}}

\def\hatx{{\hat{\x}}}

\def\Sweak{S_{weak}}

\def\betaweak{\beta_{weak}}

\def\betasec{\beta_{sec}}

\def\Ssec{S_{sec}}

\def\Sstr{S_{str}}
\def\betastr{\beta_{str}}

\newtheorem{theorem}{Theorem}

\begin{document}

\begin{singlespace}

\title {Under-determined linear systems and $\ell_q$-optimization thresholds 
}
\author{
\textsc{Mihailo Stojnic}
\\
\\
{School of Industrial Engineering}\\
{Purdue University, West Lafayette, IN 47907} \\
{e-mail: {\tt mstojnic@purdue.edu}} }
\date{}
\maketitle

\centerline{{\bf Abstract}} \vspace*{0.1in}

Recent studies of under-determined linear systems of equations with sparse solutions showed a great practical and theoretical efficiency of a particular technique called $\ell_1$-optimization. Seminal works \cite{CRT,DOnoho06CS} rigorously confirmed it for the first time. Namely, \cite{CRT,DOnoho06CS} showed, in a statistical context, that $\ell_1$ technique can recover sparse solutions of under-determined systems even when the sparsity is linearly proportional to the dimension of the system. A followup \cite{DonohoPol} then precisely characterized such a linearity through a geometric approach and a series of work\cite{StojnicCSetam09,StojnicUpper10,StojnicEquiv10} reaffirmed statements of \cite{DonohoPol} through a purely probabilistic approach. A theoretically interesting alternative to $\ell_1$ is a more general version called $\ell_q$ (with an essentially arbitrary $q$). While  $\ell_1$ is typically considered as a first available convex relaxation of sparsity norm $\ell_0$, $\ell_q,0\leq q\leq 1$, albeit non-convex, should technically be a tighter relaxation of $\ell_0$. Even though developing polynomial (or close to be polynomial) algorithms for non-convex problems is still in its initial phases one may wonder what would be the limits of an $\ell_q,0\leq q\leq 1$, relaxation even if at some point one can develop algorithms that could handle its non-convexity. A collection of answers to this and a few realted questions is precisely what we present in this paper. Namely, we look at the $\ell_q$-optimization and how it fares when used for solving under-determined linear systems with sparse solutions. Although our results are designed to be only on an introductory/conceptual level, they already hint that $\ell_q$ can in fact provide a better performance than $\ell_1$ and that designing the algorithms that would be able to handle it in a reasonable (if not polynomial) time is certainly worth further exploration.

\vspace*{0.25in} \noindent {\bf Index Terms: under-determined linear systems; sparse solutions; $\ell_q$-minimization}.

\end{singlespace}

\section{Introduction}
\label{sec:back}

In this paper we look at the under-determined linear systems of equations with sparse solutions. These systems gained a lot of attention recently in first place due to seminal results of \cite{CRT,DOnoho06CS}. In \cite{CRT,DOnoho06CS}, a particular technique called $\ell_1$ optimization was considered and it was shown in a statistical context that such a technique can recover a sparse solution (of sparsity linearly proportional to the system dimension).

To make all of this a bit more precise we start with a mathematical descriptions of linear systems. As is well known a linear system of equations can be written as
\begin{equation}
A\x=\y \label{eq:system}
\end{equation}
where $A$ is an $m\times n$ ($m<n$) system matrix and $\y$ is
an $m\times 1$ vector. Typically one is then given $A$ and $\y$ and the goal is to determine $\x$. However when ($m<n$) the odds are that there will be many solutions and that the system will be under-determined. In fact that is precisely the scenario that we will look at. However, we will slightly restrict our choice of $\y$. Namely, we will assume that $\y$ can be represented as
\begin{equation}
\y=A\tilde{\x}, \label{eq:yrepsystem}
\end{equation}
where we also assume that $\tilde{\x}$ is a $k$-sparse vector (here and in the rest of the paper, under $k$-sparse vector we assume a vector that has at most $k$ nonzero components). This essentially means that we are interested in solving (\ref{eq:system}) assuming that there is a solution that is $k$-sparse. Moreover, we will assume that there is no solution that is less than $k$-sparse, or in other words, a solution that has less than $k$ nonzero components. Such type of problems gained a lot of popularity over the last decade in first place due to their applications in a field called compressed sensing (while the literature on compressed sensing is growing on a daily basis, we here refer to two introductory papers \cite{CRT,DOnoho06CS}).

To make writing in the rest of the paper easier, we will assume the
so-called \emph{linear} regime, i.e. we will assume that $k=\beta n$
and that the number of equations is $m=\alpha n$ where
$\alpha$ and $\beta$ are constants independent of $n$ (more
on the non-linear regime, i.e. on the regime when $m$ is larger than
linearly proportional to $k$ can be found in e.g.
\cite{CoMu05,GiStTrVe06,GiStTrVe07}).

Now, given the above sparsity assumption, one can then rephrase the original problem (\ref{eq:system}) in the following way
\begin{eqnarray}
\mbox{min} & & \|\x\|_{0}\nonumber \\
\mbox{subject to} & & A\x=\y. \label{eq:l0}
\end{eqnarray}
Assuming that $\|\x\|_{0}$ counts how many nonzero components $\x$ has, (\ref{eq:l0}) is essentially looking for the sparsest $\x$ that satisfies (\ref{eq:system}), which, according to our assumptions, is exactly $\tilde{\x}$. Clearly, it would be nice if one can solve in a reasonable (say polynomial) time (\ref{eq:l0}). However, this does not appear to be easy. Instead one typically resorts to its relaxations that would be solvable in polynomial time. The first one that is typically employed is called $\ell_1$-minimization. Since what we will present in this paper will related to this technique we the following subsection provide a brief review of the $\ell_1$.

\subsection{$\ell_1$-minimization}
\label{sec:l1min}

As mentioned above, the first relaxation of (\ref{eq:l0}) that is typically employed is
the following $\ell_1$ minimization
\begin{eqnarray}
\mbox{min} & & \|\x\|_{1}\nonumber \\
\mbox{subject to} & & A\x=\y. \label{eq:l1}
\end{eqnarray}
Clearly, (\ref{eq:l1}) is an optimization problem solvable in polynomial time. Of course the question is how well does it approximate the original problem (\ref{eq:l0}). Well, for certain system dimensions it actually works very well and actually find exactly the same solution as (\ref{eq:l0}). In fact, one of the main reasons why the compressed sensing became popular is actually success of \cite{CRT,DOnoho06CS,DonohoPol} in characterizing when the solutions of (\ref{eq:l0}) and (\ref{eq:l1}) are the same. While there have been a tone of great work on $\ell_1$ we below restrict our attention to reviewing these two lines of work, in our mind, the most influential in this field.


In \cite{CRT} the authors were able to show that if
$\alpha$ and $n$ are given, $A$ is given and satisfies the restricted isometry property (RIP) (more on this property the interested reader can find in e.g. \cite{Crip,CRT,Bar,Ver,ALPTJ09}), then
any unknown vector $\tilde{\x}$ in (\ref{eq:yrepsystem}) with no more than $k=\beta n$ (where $\beta$
is a constant dependent on $\alpha$ and explicitly
calculated in \cite{CRT}) non-zero elements can be recovered by
solving (\ref{eq:l1}).

However, the RIP is only a \emph{sufficient}
condition for $\ell_1$-optimization to recover $\tilde{\x}$. Instead of characterizing $A$ through the RIP
condition, in \cite{DonohoUnsigned,DonohoPol} Donoho looked at its geometric properties/potential. Namely,
in \cite{DonohoUnsigned,DonohoPol} Donoho considered the polytope obtained by
projecting the regular $n$-dimensional cross-polytope $C_p^n$ by $A$. He then established that
the solution of (\ref{eq:l1}) will be the $k$-sparse solution of
(\ref{eq:system}) if and only if
$AC_p^n$ is centrally $k$-neighborly
(for the definitions of neighborliness, details of Donoho's approach, and related results the interested reader can consult now already classic references \cite{DonohoUnsigned,DonohoPol,DonohoSigned,DT}). In a nutshell, using the results
of \cite{PMM,AS,BorockyHenk,Ruben,VS}, it is shown in
\cite{DonohoPol}, that if $A$ is a random $m\times n$
ortho-projector matrix then with overwhelming probability $AC_p^n$ is centrally $k$-neighborly (as usual, under overwhelming probability we in this paper assume
a probability that is no more than a number exponentially decaying in $n$ away from $1$). Miraculously, \cite{DonohoPol,DonohoUnsigned} provided a precise characterization of $m$ and $k$ (in a large dimensional context) for which this happens.

It should be noted that one usually considers success of
(\ref{eq:l1}) in recovering \emph{any} given $k$-sparse $\x$ in (\ref{eq:system}). It is also of interest to consider success of
(\ref{eq:l1}) in recovering
\emph{almost any} given $\x$ in (\ref{eq:system}). We below make a distinction between these
cases and recall on some of the definitions from
\cite{DonohoPol,DT,DTciss,DTjams2010,StojnicCSetam09,StojnicICASSP09}.

Clearly, for any given constant $\alpha\leq 1$ there is a maximum
allowable value of $\beta$ such that for \emph{any} given $k$-sparse $\x$ in (\ref{eq:system}) the solution of (\ref{eq:l1})
is with overwhelming probability exactly that given $k$-sparse $\x$. One can then (as is typically done) refer to this maximum allowable value of
$\beta$ as the \emph{strong threshold} (see
\cite{DonohoPol}) and denote it as $\beta_{str}$. Similarly, for any given constant
$\alpha\leq 1$ and \emph{any} given $\x$ with a given fixed location of non-zero components and a given fixed combination of its elements signs
there will be a maximum allowable value of $\beta$ such that
(\ref{eq:l1}) finds that given $\x$ in (\ref{eq:system}) with overwhelming
probability. One can refer to this maximum allowable value of
$\beta$ as the \emph{weak threshold} and denote it by $\beta_{w}$ (see, e.g. \cite{StojnicICASSP09,StojnicCSetam09}). One can also go a step further and consider scenario where for any given constant
$\alpha\leq 1$ and \emph{any} given $\x$ with a given fixed location of non-zero components
there will be a maximum allowable value of $\beta$ such that
(\ref{eq:l1}) finds that given $\x$ in (\ref{eq:system}) with overwhelming
probability. One can then refer to such a $\beta$ as the \emph{sectional threshold} and denote it by $\beta_{sec}$ (more on the definition of the sectional threshold the interested reader can find in e.g. \cite{DonohoPol,StojnicCSetam09}).

When viewed within this frame the results of \cite{CRT,DOnoho06CS} established that $\ell_1$-minimization achieves recovery through a linear scaling of all important dimensions ($k$, $m$, and $n$). Moreover, for all $\beta$'s defined above lower bounds were provided in \cite{CRT}. On the other hand, the results of \cite{DonohoPol,DonohoUnsigned} established the exact values of $\beta_w$ and provided lower bounds on $\beta_{str}$ and $\beta_{sec}$.

In a series of our own work (see, e.g. \cite{StojnicICASSP09,StojnicCSetam09,StojnicUpper10}) we then created an alternative probabilistic approach which was capable of providing the precise characterization of $\beta_w$ as well and thereby reestablishing the results of Donoho \cite{DonohoPol} through a purely probabilistic approach. We also presented in \cite{StojnicCSetam09} further results related to lower bounds on $\beta_{str}$ and $\beta_{sec}$.

Of course, there are many other algorithms that can be used to attack (\ref{eq:l0}). Among them are also numerous variations of the standard $\ell_1$-optimization from e.g. \cite{CWBreweighted,SChretien08,SaZh08,StojnicICASSP10knownsupp} as well as many other conceptually completely different ones from e.g. \cite{JATGomp,JAT,NeVe07,DTDSomp,NT08,DaiMil08,DonMalMon09}. While all of them are fairly successful in their own way and with respect to various types of performance measure, one of them, namely the so called AMP from \cite{DonMalMon09}, is of particular interest when it comes to $\ell_1$. What is fascinating about AMP is that it is a fairly fast algorithm (it does require a bit of tuning though) and it has provably the same statistical performance as (\ref{eq:l1}) (for more details on this see, e.g. \cite{DonMalMon09,BayMon10}). Since our main goal in this paper is to a large degree related to $\ell_1$ we stop short of reviewing further various alternatives to (\ref{eq:l1}) and instead refer to any of the above mentioned papers as well as our own \cite{StojnicCSetam09,StojnicUpper10} where these alternatives were revisited in a bit more detail.

Below, we instead switch to a further modification of $\ell_1$ called $\ell_q$ that will be the main subject of this paper.

\subsection{$\ell_q$-minimization}
\label{sec:lqmin}

As mentioned above, the first relaxation of (\ref{eq:l0}) that is typically employed is
the $\ell_1$ minimization from (\ref{eq:l1}). The reason for that is that it is the first of the norm relaxations that results in an optimization problem that is solvable in polynomial time. One can alternatively look at the following (tighter) relaxation (considered in e.g. \cite{GN03,GN04,GN07,FL08})
\begin{eqnarray}
\mbox{min} & & \|\x\|_{q}\nonumber \\
\mbox{subject to} & & A\x=\y. \label{eq:lq}
\end{eqnarray}
We will for concreteness assume $q\in[0,1]$; however, we do mention that when it comes to our own results that we will present below there is really no need for such a restriction, i.e. our results can easily be adapted to work for a wider range of $q$. Clearly, (\ref{eq:lq}) is an optimization problem which is not known to be solvable in polynomial time. Moreover, developing fast algorithms to solve it is a fairly attractive area of research. Since our goal will be recovering abilities of (\ref{eq:lq}) rather than how it can be solved we don't analyze in further details practical algorithmic aspects of (\ref{eq:lq}). In other words, we will assume that (\ref{eq:lq}) somehow can be solved and then we will look at scenarios when such a solution matches $\tilde{\x}$. In a way our analysis will provide some answers to question: if one can solve (\ref{eq:lq}) in a reasonable (if not polynomial) amount of time how likely is that its solution will be $\tilde{\x}$.

Of course, this is the same type of question we considered when discussing performance of (\ref{eq:l1}) above and obviously the same type of question attacked in \cite{CRT,DOnoho06CS,DonohoPol,StojnicCSetam09,StojnicUpper10}. To be a bit more specific, one can then ask for what system dimensions (\ref{eq:lq}) actually works well and finds exactly the same solution as (\ref{eq:l0}), i.e. $\tilde{\x}$. A typical way to attack such a question would be to translate the results that relate to $\ell_1$ to general $\ell_q$ case. In fact that is exactly what has been done for many techniques, including obviously the RIP one developed in \cite{CRT}. In this paper, we will attempt to translate our own results from \cite{StojnicCSetam09}. To that end, we will present results that relate to the sectional, strong, and weak thresholds of $\ell_q$ minimization. The definitions of these thresholds will follow the above introduced definitions for $\ell_1$-thresholds with a very few minor modifications. We will introduce them throughout the paper as we need them.

We organize the rest of the paper in the following way. In Section
\ref{sec:secthr} we present the core of the mechanism and how it can be used to obtain the sectional thresholds for $\ell_q$ minimization. In Section \ref{sec:strthr} we will then present a neat modification of the mechanism so that it can handle the strong thresholds as well. In Section \ref{sec:weakthr} we present the weak thresholds results. In Section \ref{sec:conc} we discuss obtained results and provide several conclusions related to their importance.

\section{$\ell_q$-minimization sectional threshold}
\label{sec:secthr}

In this section we start assessing the performance of $\ell_q$ minimization by looking at its sectional thresholds. Before proceeding further we slightly readjust the definition of the $\ell_1$ sectional thresholds given above so that it fits the $\ell_q$ case considered here. Namely, one considers a scenario where for any given constant
$\alpha\leq 1$ and \emph{any} $\tilde{\x}$ in \ref{eq:yrepsystem} with a given fixed location of non-zero components
there will be a maximum allowable value of $\beta$ such that the solution of
(\ref{eq:lq}) is that given $\tilde{\x}$ with overwhelming
probability. We will refer to such a $\beta$ as the \emph{sectional threshold} and will denote it by $\beta_{sec}^{(q)}$ (we again recall that more on the definition of the sectional threshold the interested reader can find in e.g. \cite{DonohoPol,StojnicCSetam09}).

\subsection{Sectional threshold preliminaries}
\label{sec:secthrprelim}

Below we will provide a way to quantify behavior of $\beta_{sec}^{(q)}$. In doing so we will rely on some of the mechanisms presented in \cite{StojnicCSetam09}.
and along the same lines will assume a substantial level of familiarity with many of the well-known results that relate to the performance characterization of (\ref{eq:l1}) (we will fairly often recall on many results/definitions that we established in \cite{StojnicCSetam09}). We start by introducing a nice way of characterizing sectional success/failure of (\ref{eq:lq}).

\begin{theorem}(Nonzero part of $\x$ has fixed location)
Assume that an $m\times n$ matrix $A$ is given. Let $\tilde{X}_{sec}$ be the collection of all $k$-sparse vectors $\tilde{\x}$ in $R^n$ for which $\tilde{\x}_1=\tilde{\x}_2=\dots=\tilde{\x}_{n-k}=0$. Let $\tilde{\x}^{(i)}$ be any $k$-sparse vector from $\tilde{X}_{sec}$. Further, assume that $\y^{(i)}=A\tilde{\x}^{(i)}$ and that $\w$ is
an $n\times 1$ vector. If
\begin{equation}
(\forall \w\in \textbf{R}^n | A\w=0) \quad  \sum_{i=n-k+1}^n |\w_i|^q<\sum_{i=1}^{n-k}|\w_{i}|^q
\label{eq:thmeqgensec1}
\end{equation}
then the solution of (\ref{eq:lq}) for every pair $(\y^{(i)},A)$ is the corresponding $\tilde{\x}^{(i)}$.
\label{thm:thmgensec}
\end{theorem}
\begin{proof}
The proof follows directly from the corresponding results for $\ell_1$ (see, e.g. Theorem $2$ in \cite{StojnicICASSP09} and references therein). For the completeness we just sketch the argument again. Let $\hatx$ be the solution of (\ref{eq:lq}). We want to show that if (\ref{eq:thmeqgensec1}) holds then $\hatx=\tilde{\x}$. To that end assume opposite, i.e. assume that (\ref{eq:thmeqgensec1}) holds but $\hatx\neq\tilde{\x}$. Then since $\y=A\hatx$ and $\y=A\tilde{\x}$ one must have $\hatx =\tilde{\x}+\w$ with $\w$ such that $A\w=0$. Also, since $\hatx$ is the solution of (\ref{eq:lq}) one has that
\begin{equation}
\sum_{i=1}^n|\x_i+\w_i|^q\leq \sum_{i=1}^{n}|\x_i|^q.\label{eq:absval}
\end{equation}
Then the following must hold as well
\begin{equation}
\sum_{i=1}^{n-k} |\w_i|^q-\sum_{i=n-k+1}^n |\w_i|^q\leq 0.\label{eq:wcon}
\end{equation}
or equivalently
\begin{equation}
\sum_{i=1}^{n-k} |\w_i|^q\leq\sum_{i=n-k+1}^n |\w_i|^q.\label{eq:wcon1}
\end{equation}
Clearly, (\ref{eq:wcon1}) contradicts (\ref{eq:thmeqgensec1}) and $\hatx\neq\tilde{\x}$ can not hold. Therefore $\hatx=\tilde{\x}$ which is exactly what the theorem claims.
\end{proof}

\noindent \textbf{Remark:} The above proof is not our own. If nothing else it directly follows the strategy that would be applied for $q=1$, i.e. $\ell_1$ which had been detailed in many places, see e.g. \cite{DH01,FN,LN,Y,XHapp,SPH,DTbern}. Moreover, such a strategy has already been applied to this very same case of general $q$ as well, see e.g. \cite{GN03,GN04,GN07,FL08}. As we just mentioned, the above proof is not our own and we presented its a sketch just for the completeness. Also, although we did not emphasize it in the above theorem, we mention here that the condition given in the theorem is not only sufficient to characterize sectional equivalence of (\ref{eq:l0}) and (\ref{eq:lq}) but it is also necessary.

We then, following the methodology of \cite{StojnicCSetam09},
start by defining a set $\Ssec$
\begin{equation}
\Ssec=\{\w\in S^{n-1}| \quad \sum_{i=n-k+1}^n |\w_i|^q\geq \sum_{i=1}^{n-k}|\w_{i}|^q\},\label{eq:defSsec}
\end{equation}
where $S^{n-1}$ is the unit sphere in $R^n$. The methodology of \cite{StojnicCSetam09} then invokes the following classic result of Gordon (the version below is a slightly modified version of Gordon's original formulation).
\begin{theorem}(\cite{Gordon88} Escape through a mesh)
\label{thm:Gordonmesh} Let $S$ be a subset of the unit Euclidean
sphere $S^{n-1}$ in $R^{n}$. Let $Y$ be a random
$(n-m)$-dimensional subspace of $R^{n}$, spanned by $(n-m)$ vectors from $R^n$ with i.i.d. standard normal components. Let
\begin{equation}
w_D(S)=E\sup_{\w\in S} (\h^T\w) \label{eq:widthdef}
\end{equation}
where $\h$ is a random column vector in $R^{n}$ with i.i.d. standard normal components. Assume that
$w_D(S)<\left ( \sqrt{m}-\frac{1}{4\sqrt{m}}\right )$. Then
\begin{equation}
P(Y\cap S=0)>1-3.5e^{-\frac{\left (
\sqrt{m}-\frac{1}{4\sqrt{m}}-w_D(S) \right ) ^2}{18}}.
\label{eq:thmesh}
\end{equation}
\end{theorem}
\textbf{Remark}: Gordon's original constant $3.5$ was substituted by
$2.5$ in \cite{RVmesh}. Both constants are not subject of our detailed considerations. However, we do mention in passing that to the best of our knowledge it is an open problem to determine the exact value of this constant as well as to improve and ultimately determine the exact value as well of somewhat high constant $18$.

The methodology of \cite{StojnicCSetam09} then proceeds by characterizing
\begin{equation}
w_D(\Ssec)=E\max_{\w\in \Ssec} (\h^T\w),\label{eq:negham1}
\end{equation}
where to facilitate the exposition we replace $\sup$ with a $\max$. Below we present a way to create an upper-bound on $w_D(\Ssec)$. Equalling such an upper bound with $\sqrt{m}$ would be roughly enough to provide a characterization of the sectional thresholds.

\subsection{Sectional threshold computation}
\label{sec:secthrcomp}

Let $f(\w)=\h^T\w$ and
we start with the following line of identities
\begin{multline}
\hspace{-.5in}\max_{\w\in\Ssec}f(\w)=-\min_{\w\in\Ssec} -\h^T\w=-\min_{\w}\max_{\gamma_{sec}\geq 0,\nu_{sec}\geq 0} -\h^T\w
-\nu_{sec}\sum_{i=n-k+1}^{n}|\w_i|^q
+\nu_{sec}\sum_{i=1}^{n-k}|\w_i|^q+\gamma_{sec}\sum_{i=1}^{n}\w_i^2-\gamma_{sec}\\
\leq -\max_{\gamma_{sec}\geq 0,\nu_{sec}\geq 0}\min_{\w} -\h^T\w
-\nu_{sec}\sum_{i=n-k+1}^{n}|\w_i|^q
+\nu_{sec}\sum_{i=1}^{n-k}|\w_i|^q+\gamma_{sec}\sum_{i=1}^{n}\w_i^2-\gamma_{sec}\\
=-\max_{\gamma_{sec}\geq 0,\nu_{sec}\geq 0}\min_{\w} -\sum_{i=n-k+1}^{n}(|\h_i||\w_i|+\nu_{sec}|\w_i|^q)
+\sum_{i=1}^{n-k}(-|\h_i||\w_i|+\nu_{sec}|\w_i|^q)+\gamma_{sec}\sum_{i=1}^{n}\w_i^2-\gamma_{sec}\\
=\min_{\gamma_{sec}\geq 0,\nu_{sec}\geq 0}\max_{\w} \sum_{i=n-k+1}^{n}(|\h_i||\w_i|+\nu_{sec}|\w_i|^q)
+\sum_{i=1}^{n-k}(|\h_i||\w_i|-\nu_{sec}|\w_i|^q)-\gamma_{sec}\sum_{i=1}^{n}\w_i^2+\gamma_{sec}\\
=\min_{\gamma_{sec}\geq 0,\nu_{sec}\geq 0} f_1(q,\h,\nu_{sec},\gamma_{sec},\beta)+\gamma_{sec},\label{eq:seceq1}
\end{multline}
where
\begin{equation}
f_1(q,\h,\nu_{sec},\gamma_{sec},\beta)=\max_{\w}\left (\sum_{i=n-k+1}^{n}(|\h_i||\w_i|+\nu_{sec}|\w_i|^q-\gamma_{sec}\w_i^2)
+\sum_{i=1}^{n-k}(|\h_i||\w_i|-\nu_{sec}|\w_i|^q-\gamma_{sec}\w_i^2)\right ).\label{eq:deff1}
\end{equation}
One then has
\begin{multline}
w_D(\Ssec)=E\max_{\w\in\Ssec}\h^T\w=E\max_{\w\in\Ssec}f(\w)=
E\min_{\gamma_{sec}\geq 0,\nu_{sec}\geq 0} f_1(q,\h,\nu_{sec},\gamma_{sec},\beta)+\gamma_{sec}\\
\leq \min_{\gamma_{sec}\geq 0,\nu_{sec}\geq 0} E f_1(q,\h,\nu_{sec},\gamma_{sec},\beta)+\gamma_{sec}.\label{eq:wdineq}
\end{multline}
Now if one sets $\w_{i}=\frac{\w_{i}^{(s)}}{\sqrt{n}}$, $\gamma_{sec}=\gamma_{sec}^{(s)}\sqrt{n}$, and $\nu_{sec}=\nu_{sec}^{(s)}\sqrt{n}^{q-1}$ (where $\w_{i}^{(s)}$, $\gamma_{sec}^{(s)}$, and $\nu_{sec}^{(s)}$ are independent of $n$) then (\ref{eq:wdineq}) gives
\begin{multline}
\lim_{n\rightarrow\infty}\frac{w_D(\Ssec)}{\sqrt{n}}=\lim_{n\rightarrow\infty}\frac{E\max_{\w\in\Ssec}\h^T\w}{\sqrt{n}}
=\lim_{n\rightarrow\infty}\frac{E\max_{\w\in\Ssec}f(\w)}{\sqrt{n}}\\
\hspace{-.3in}=\lim_{n\rightarrow\infty}\frac{E\min_{\gamma_{sec}\geq 0,\nu_{sec}\geq 0} (f_1(q,\h,\nu_{sec},\gamma_{sec},\beta)+\gamma_{sec})}{\sqrt{n}}
\leq \lim_{n\rightarrow\infty}\frac{\min_{\gamma_{sec}\geq 0,\nu_{sec}\geq 0} (E f_1(q,\h,\nu_{sec},\gamma_{sec},\beta)+\gamma_{sec})}{\sqrt{n}}\\
\hspace{-.5in}=\min_{\gamma_{sec}^{(s)}\geq 0,\nu_{sec}^{(s)}\geq 0}  ((\beta E\max_{\w_i^{(s)}}(|\h_i||\w_i^{(s)}|+\nu_{sec}^{(s)}|\w_i^{(s)}|^q-\gamma_{sec}^{(s)}(\w_i^{(s)})^2)\\
+(1-\beta)E\max_{\w_j^{(s)}}(|\h_j||\w_j^{(s)}|-\nu_{sec}^{(s)}|\w_j^{(s)}|^q-\gamma_{sec}^{(s)}(\w_j^{(s)})^2) )+\gamma_{sec}^{(s)} )
=\min_{\gamma_{sec}^{(s)}\geq 0,\nu_{sec}^{(s)}\geq 0} \left (\left (\beta I_{sec}^{(1)}
+(1-\beta)I_{sec}^{(2)}\right )+\gamma_{sec}^{(s)}\right ),\label{eq:wdineq1}
\end{multline}
where
\begin{eqnarray}
I_{sec}^{(1)} & = & E\max_{\w_i^{(s)}}(|\h_i||\w_i^{(s)}|+\nu_{sec}^{(s)}|\w_i^{(s)}|^q-\gamma_{sec}^{(s)}(\w_i^{(s)})^2)\nonumber \\
I_{sec}^{(2)} & = & E\max_{\w_j^{(s)}}(|\h_j||\w_j^{(s)}|-\nu_{sec}^{(s)}|\w_j^{(s)}|^q-\gamma_{sec}^{(s)}(\w_j^{(s)})^2).\label{eq:defI1I2sec}
\end{eqnarray}
We summarize the above results related to the sectional threshold ($\beta_{sec}^{(q)}$) in the following theorem.

\begin{theorem}(Sectional threshold - lower bound)
Let $A$ be an $m\times n$ measurement matrix in (\ref{eq:system})
with i.i.d. standard normal components. Let $\tilde{X}_{sec}$ be the collection of all $k$-sparse vectors $\tilde{\x}$ in $R^n$ for which $\tilde{\x}_1=0,\tilde{\x}_2=0,,\dots,\tilde{\x}_{n-k}=0$. Let $\tilde{\x}^{(i)}$ be any $k$-sparse vector from $\tilde{X}_{sec}$. Further, assume that $\y^{(i)}=A\tilde{\x}^{(i)}$. Let $k,m,n$ be large
and let $\alpha=\frac{m}{n}$ and $\betasec^{(q)}=\frac{k}{n}$ be constants
independent of $m$ and $n$. Let
\begin{eqnarray}
I_{sec}^{(1)} & = & E\max_{\w_i}(|\h_i||\w_i^{(s)}|+\nu_{sec}^{(s)}|\w_i^{(s)}|^q-\gamma_{sec}^{(s)}(\w_i^{(s)})^2)\nonumber \\
I_{sec}^{(2)} & = & E\max_{\w_j}(|\h_j||\w_j^{(s)}|-\nu_{sec}^{(s)}|\w_j^{(s)}|^q-\gamma_{sec}^{(s)}(\w_j^{(s)})^2).\label{eq:defI1I2secthm}
\end{eqnarray}
If $\alpha$ and $\betasec^{(q)}$ are such that
\begin{equation}
\min_{\gamma_{sec}^{(s)}\geq 0,\nu_{sec}^{(s)}\geq 0} \left (\left (\beta_{sec}^{(q)} I_{sec}^{(1)}
+(1-\beta_{sec}^{(q)})I_{sec}^{(2)}\right )+\gamma_{sec}^{(s)}\right )<\sqrt{\alpha},\label{eq:seccondthmsec}
\end{equation}
then with overwhelming probability the solution of (\ref{eq:lq}) for every pair $(\y^{(i)},A)$ is the corresponding $\tilde{\x}^{(i)}$.\label{thm:thmsecthrlq}
\end{theorem}
\begin{proof}
Follows from the above discussion.
\end{proof}

The results for the sectional threshold obtained from the above theorem
are presented in Figure \ref{fig:sec}. To be a bit more specific, we selected four different values of $q$, namely $q\in\{0,0.1,0.3,0.5\}$ in addition to standard $q=1$ case already discussed in \cite{StojnicCSetam09}.
\begin{figure}[htb]
\centering
\centerline{\epsfig{figure=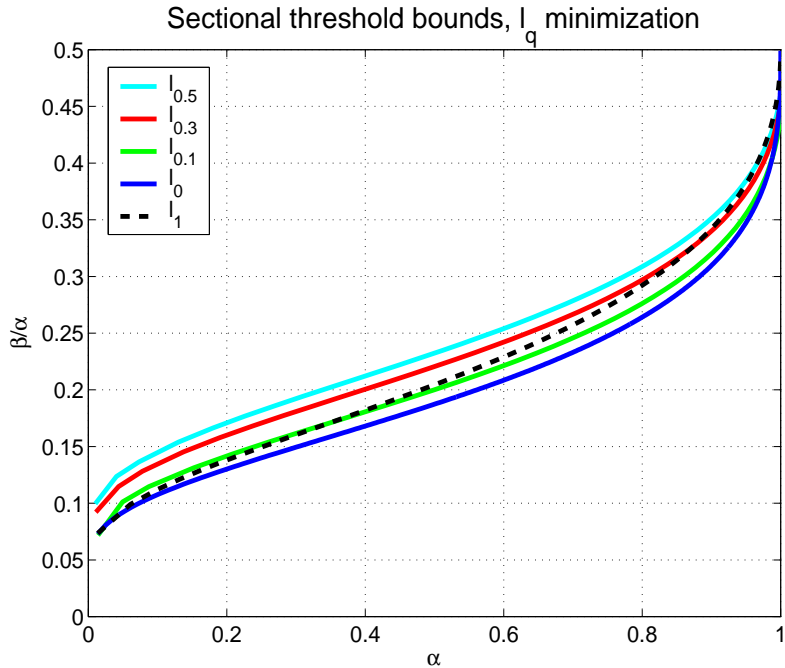,width=10.5cm,height=9cm}}
\caption{\emph{Sectional} threshold, $\ell_q$-optimization}
\label{fig:sec}
\end{figure}
As can be seen from Figure \ref{fig:sec}, for some values of $q$ the results are better than for $q=1$. However, for some the results are worse. Of course one has to be careful how to interpret this. First, one may naturally expect that as $q$ goes down the threshold results become better, i.e. the resulting curves go up. That does happen down to some values for $q$; however, after that the curves start sliding down and eventually for $q=0$ we actually have a curve that is even below $q=1$ case. Of course this just shows that our methodology works successfully to a degree, i.e. its a lower-bounding tendency eventually comes into a full effect. Of, course if one is interested in the best possible sectional threshold values for any $q$ rather than the methodology itself the curves that go down as $q$ goes up could be ignored. However, we kept them on the plot to emphasize that the proposed methodology has some inherent deficiencies.

The obtained results can also be compared with the best known ones for $\ell_1$-minimization from \cite{StojnicLiftStrSec13} as well. However, since these are fairly close to the curve that corresponds to $\ell_1$ given in Figure \ref{fig:sec} we skip adding these plots and making the figure even more detailed.

Also, all results are obtained after numerical computations. They mostly included numerical optimizations which were all (except maximization over $\w$) done on a local optimum level. We do not know how (if in any way) solving them on a global optimum level would affect the location of the plotted curves. Also, numerical integrations were done on a finite precision level as well which could have potentially harmed the final results as well. Still, we believe that the methodology can not achieve substantially more than what we presented in Figure \ref{fig:sec} (and hopefully is not severely degraded with numerical integrations and maximization over $\w$).

Solving over $\nu_{sec}^{(s)}$ and $\gamma_{sec}^{(s)}$ on a local optimum level may lower the curves but it
certainly does not jeopardize their lower bounding rigorousness. However, solving the maximization over $\w$,
even on a global optimum level as we did, may do so. Since this may jeopardize the lower bounding rigorousness
in addition to plots in Figure \ref{fig:sec} we present in Tables \ref{tab:sectab1}, \ref{tab:sectab2}, and \ref{tab:sectab3} the concrete values we obtained for $\nu_{sec}^{(s)}$ and
$\gamma_{sec}^{(s)}$ for certain $\beta_{sec}^{(q)}$ on the way to computing corresponding $\alpha$ (as indicated above the tables, Table \ref{tab:sectab1} contains data for $\ell_q, q=0.5$, Table1 \ref{tab:sectab2} contains data for $\ell_q, q=0.3$, and Table \ref{tab:sectab3} contains data for $\ell_q, q=0.1$,). That way the
interested reader can double check if the optimization over $\w$ in any way endangered the lower-bounding rigorousness.
Of course, we do reemphasize that the results presented in the above theorem are completely rigorous,
it is just that some of the numerical work that we performed could have been a bit imprecise
(we firmly believe that this is not the case; however with finite numerical precision one has to be cautious all the time).

\subsection{Special cases}
\label{sec:secthrspecial}

In this subsection we briefly note that some of the above computations can be done in a faster, more explicit fashion.

\subsubsection{$q\rightarrow 0$}
\label{sec:secthrspecialq0}

The first case we consider is $q=0$. From the plot given in Figure \ref{fig:sec} the methodology is not quite successful for this case. Nevertheless, the curve given in Figure \ref{fig:sec} can be obtained in a more direct fashion without all the computations required by Theorem \ref{thm:thmsecthrlq}. Here is a brief sketch how one can proceed. Let
\begin{equation}
\tilde{\h}=[\h_{(1)},\h_{(2)},\dots,\h_{(n-k)},|\h_{n-k+1}|,|\h_{n-k+2}|,\dots,|\h_n|],\label{eq:defhq0}
\end{equation}
where
$[\h_{(1)},\h_{(2)},\dots,\h_{(n-k)}]$ are the absolute values of components of $[\h_{1},\h_{2},\dots,\h_{n-k}]$ sorted in an increasing order. Then one has
\begin{multline}
\lim_{n\rightarrow\infty}\frac{w_D(\Ssec)}{\sqrt{n}}=\lim_{n\rightarrow\infty}\frac{E\max_{\w\in\Ssec}\h^T\w}{\sqrt{n}}
=\lim_{n\rightarrow\infty}\frac{E\max_{\w\in\Ssec}\sum_{i=1}^{n}\tilde{\h}_i|\w_i|}{\sqrt{n}}\\
=\lim_{n\rightarrow\infty}\frac{E\sqrt{\sum_{i=n-2k+1}^{n-k}\h_{(i)}^2+\sum_{i=n-k+1}^{n}\h_{i}^2}}{\sqrt{n}}
\leq \lim_{n\rightarrow\infty}\frac{\sqrt{E\sum_{i=n-2k+1}^{n-k}\h_{(i)}^2+E\sum_{i=n-k+1}^{n}\h_{i}^2}}{\sqrt{n}}.\label{eq:wdlq0}
\end{multline}
Applying the machinery of \cite{StojnicCSetam09} then gives
\begin{multline}
\lim_{n\rightarrow\infty}\frac{w_D(\Ssec)}{\sqrt{n}}
\leq \lim_{n\rightarrow\infty}\frac{\sqrt{E\sum_{i=n-2k+1}^{n-k}\h_{(i)}^2+E\sum_{i=n-k+1}^{n}\h_{i}^2}}{\sqrt{n}}\\
=\sqrt{\beta_{sec}^{(0)}+(1-\beta_{sec}^{(0)})\frac{2}{\sqrt{\pi}}\mbox{erfinv}\left (\frac{1-2\beta_{sec}^{(0)}}{1-\beta_{sec}^{(0)}}\right )
e^{-\left (\mbox{erfinv}\left (\frac{1-2\beta_{sec}^{(0)}}{1-\beta_{sec}^{(0)}}\right)\right )^2}}.\label{eq:wdlq0}
\end{multline}
Equalling the quantity on the right hand side with $\sqrt{\alpha}$ then gives the characterization of $\ell_0$ curve in Figure \ref{fig:sec}.

\begin{table}
\caption{Sectional threshold bounds $\ell_q,q=0.5$}\vspace{0in}
\centering\hspace{-.3in}
\begin{tabular}{||c|c|c|c|c|c|c|c|c|c|c|c||}\hline\hline
$\beta_{sec}^{(q)}$  & $0.0050$ &  $0.0200$ & $0.0400$ & $0.0600$ & $0.0900$ & $0.1200$ & $0.1500$ & $0.2000$ & $0.2500$ & $0.3200$ & $0.4500$  \\ \hline\hline
$\alpha$             & $0.0405$ &  $0.1299$ & $0.2262$ & $0.3091$ & $0.4173$ & $0.5112$ & $0.5938$ & $0.7105$ & $0.8051$ & $0.9046$ & $0.9974$  \\ \hline
$\nu_{sec}^{(s)}$    & $5.8112$ &  $3.2935$ & $2.3730$ & $1.9152$ & $1.5033$ & $1.2328$ & $1.0329$ & $0.7910$ & $0.6021$ & $0.3906$ & $0.0866$  \\ \hline
$\gamma_{sec}^{(s)}$ & $0.1005$ &  $0.1800$ & $0.2372$ & $0.2775$ & $0.3222$ & $0.3565$ & $0.3841$ & $0.4199$ & $0.4475$ & $0.4740$ & $0.4977$  \\ \hline\hline
\end{tabular}
\label{tab:sectab1}
\end{table}

\begin{table}
\caption{Sectional threshold bounds $\ell_q,q=0.3$}\vspace{0in}
\hspace{-0in}\centering
\begin{tabular}{||c|c|c|c|c|c|c|c|c|c|c|c||}\hline\hline
$\beta_{sec}^{(q)}$  & $0.0050$ & $0.0100$ & $0.0300$ & $0.0500$ & $0.0800$ & $0.1100$ & $0.1500$ & $0.1900$ & $0.2400$ & $0.3100$ & $0.4500$  \\ \hline\hline
$\alpha$             & $0.0436$ & $0.0780$ & $0.1900$ & $0.2821$ & $0.3992$ & $0.4991$ & $0.6124$ & $0.7073$ & $0.8042$ & $0.9047$ & $0.9992$  \\ \hline
$\nu_{sec}^{(s)}$    & $9.2019$ & $6.4961$ & $3.5738$ & $2.6101$ & $1.8927$ & $1.4778$ & $1.1231$ & $0.8727$ & $0.6335$ & $0.3965$ & $0.0667$\\ \hline
$\gamma_{sec}^{(s)}$ & $0.1039$ & $0.1398$ & $0.2174$ & $0.2649$ & $0.3152$ & $0.3520$ & $0.3900$ & $0.4192$ & $0.4471$ & $0.4741$ & $0.4983$ \\ \hline\hline
\end{tabular}
\label{tab:sectab2}
\end{table}

\begin{table}
\caption{Sectional threshold bounds $\ell_q,q=0.1$}\vspace{0in}
\hspace{-0in}\centering
\begin{tabular}{||c|c|c|c|c|c|c|c|c|c|c|c||}\hline\hline
$\beta_{sec}^{(q)}$  & $0.0010$ & $0.0100$ & $0.0300$ & $0.0500$ & $0.0700$ & $0.1000$ & $0.1300$ & $0.1700$ & $0.2200$ & $0.2900$ & $0.4400$ \\ \hline\hline
$\alpha$             & $0.0139$ & $0.0873$ & $0.2089$ & $0.3069$ & $0.3912$ & $0.4998$ & $0.5921$ & $0.6953$ & $0.7983$ & $0.9023$ & $0.9997$ \\ \hline
$\nu_{sec}^{(s)}$    & $26.050$ & $9.2658$ & $4.5185$ & $3.1043$ & $2.3942$ & $1.7389$ & $1.3434$ & $0.9913$ & $0.6908$ & $0.4044$ & $0.0514$\\ \hline
$\gamma_{sec}^{(s)}$ & $0.0781$ & $0.1473$ & $0.2282$ & $0.2764$ & $0.3119$ & $0.3528$ & $0.3830$ & $0.4153$ & $0.4453$ & $0.4734$ & $0.4983$   \\ \hline\hline
\end{tabular}
\label{tab:sectab3}
\end{table}

\subsubsection{$q=\frac{1}{2}$}
\label{sec:secthrspecialq05}

Another special case that allows a further simplification of the results presented in Theorem \ref{thm:thmsecthrlq} is when $q=\frac{1}{2}$. In this case one can be more explicit when it comes to the optimization over $\w$. Namely, taking simply the derivatives one finds
\begin{equation*}
|\h_i|\pm q\nu_{sec}^{(s)}|\w_i^{(s)}|^{q-1}-2\gamma_{sec}^{(s)}|\w_i^{(s)}|=0,
\end{equation*}
which when $q=\frac{1}{2}$ gives
\begin{eqnarray}
& & |\h_i|\pm\frac{1}{2}\nu_{sec}^{(s)}|\w_i^{(s)}|^{-1/2}-2\gamma_{sec}^{(s)}|\w_i^{(s)}|=0\nonumber \\
& \Leftrightarrow & |\h_i|\sqrt{|\w_i^{(s)}|}\pm\frac{1}{2}\nu_{sec}^{(s)}-2\gamma_{sec}^{(s)}\sqrt{|\w_i^{(s)}|}^{3}=0,\label{eq:cubicq05}
\end{eqnarray}
which is a cubic equation and can be solved explicitly. This of course substantially facilitates the integrations over $\h_i$. Also, similar strategy can be applied for other rational $q$. However, the ``explicit" solutions soon become more complicated than the numerical ones and we skip presenting them.

\section{$\ell_q$-minimization strong threshold}
\label{sec:strthr}

In this section we present results related to the $\ell_q$ minimization strong thresholds. As was the case in the previous section, before proceeding further we slightly readjust the definition of the $\ell_1$ strong thresholds given earlier in the context of $\ell_1$ minimization so that it fits the $\ell_q$ case considered here. Namely, one considers a scenario where for any given constant
$\alpha\leq 1$ and \emph{any} $\tilde{\x}$ in \ref{eq:yrepsystem} with a given fixed location of non-zero components
there will be a maximum allowable value of $\beta$ such that the solution of
(\ref{eq:lq}) is that given $\tilde{\x}$ with overwhelming
probability. We will refer to such a $\beta$ as the \emph{strong threshold} and will denote it by $\beta_{sec}^{(q)}$ (we again recall that more on the definition of the strong threshold the interested reader can find in e.g. \cite{DonohoPol,StojnicCSetam09}).

\subsection{Strong threshold preliminaries}
\label{sec:strthrprelim}

Below we will provide a way to quantify behavior of $\beta_{sec}^{(q)}$. In doing so we will, as in the previous section, rely on some of the mechanisms presented in \cite{StojnicCSetam09} and a few additional ones from \cite{StojnicLiftStrSec13}. Along the same lines, we will assume a substantial level of familiarity with many of the well-known results that relate to the performance characterization of (\ref{eq:l1}) (we will fairly often recall on many results/definitions that we established in \cite{StojnicCSetam09,StojnicLiftStrSec13}). We start by introducing a nice way of characterizing strong success/failure of (\ref{eq:lq}).

\begin{theorem}(Any $k$-sparse $\x$)
Assume that an $m\times n$ matrix $A$ is given. Let $\tilde{X}_{str}$ be the collection of all $k$-sparse vectors in $R^n$. Let $\tilde{\x}^{(i)}$ be any $k$-sparse vector from $\tilde{X}$. Further, assume that $\y^{(i)}=A\tilde{\x}^{(i)}$ and that $\w$ is
an $n\times 1$ vector. If
\begin{equation}
(\forall \w\in \textbf{R}^n | A\w=0) \quad  \sum_{i=1}^n \b_i |\w_i|^q>0,\sum_{i=1}^n\b_i=2n-k,\b_i^2=1),\label{eq:thmeqgensec1}
\end{equation}
then the solution of (\ref{eq:lq}) for every pair $(\y^{(i)},A)$ is the corresponding $\tilde{\x}^{(i)}$.
\label{thm:thmgenstr}
\end{theorem}
\begin{proof}
The proof follows directly from Theorem \ref{thm:thmgenstr} by considering all different locations of $k$ nonzero components of $\tilde{\x}$.
As such, it obviously follows directly from the corresponding results for $\ell_q$ (see, e.g. Theorem $2$ in \cite{StojnicICASSP09} and references therein).
\end{proof}

\noindent \textbf{Remark:} As mentioned after the corresponding sectional threshold theorem, the above theorem is not our own. It clearly follows from the strategy that would be applied for $q=1$, i.e. $\ell_1$ which had been detailed in many places, see e.g. \cite{DH01,FN,LN,Y,XHapp,SPH,DTbern}. Also, as mentioned earlier, such a strategy has already been adapted to this very same case of general $q$ as well, see e.g. \cite{GN03,GN04,GN07,FL08}.

We now start by following what we did in the previous section and essentially in \cite{StojnicCSetam09}.
Let $\Sstr$ be the following set
\begin{equation}
\Sstr=\{\w\in S^{n-1}| \quad \sum_{i=1}^n \b_i|\w_i|^q\geq 0, \sum_{i=1}^n\b_i=2n-k,\b_i^2=1\},\label{eq:defSstr}
\end{equation}
where $S^{n-1}$ is the unit sphere in $R^n$. The methodology of the previous section and \cite{StojnicCSetam09} then proceeds by characterizing
\begin{equation}
w_D(\Sstr)=E\max_{\w\in \Sstr} (\h^T\w),\label{eq:negham1str}
\end{equation}
where, as in previous section, to facilitate the exposition we replace $\sup$ with a $\max$. Below we present a way to create an upper-bound on $w_D(\Sstr)$. Equalling such an upper bound with $\sqrt{m}$ would be roughly enough to provide a characterization of the strong thresholds.

\subsection{Strong threshold computation}
\label{sec:strthrcomp}

As earlier, let $f(\w)=\h^T\w$ and
we start with the following line of identities
\begin{multline}
\hspace{-.5in}\max_{\w\in\Sstr}f(\w)\\
=-\min_{\w\in\Sstr} -\h^T\w=-\min_{\w,\sum_{i=1}^n\b_i=2n-k,\b_i^2=1}\max_{\gamma_{str}\geq 0,\nu_{str}\geq 0} -\h^T\w
-\nu_{str}\sum_{i=1}^{n}\b_i|\w_i|^q+\gamma_{str}\sum_{i=1}^{n}\w_i^2-\gamma_{str}\\
\leq -\max_{\gamma_{str}\geq 0,\nu_{str}\geq 0}\min_{\w,\sum_{i=1}^n\b_i=2n-k,\b_i^2=1} -\h^T\w
-\nu_{str}\sum_{i=1}^{n}\b_i|\w_i|^q
+\gamma_{str}\sum_{i=1}^{n}\w_i^2-\gamma_{str}\\
= \min_{\gamma_{str}\geq 0,\nu_{str}\geq 0}\max_{\w,\sum_{i=1}^n\b_i=2n-k,\b_i^2=1} \h^T\w
+\nu_{str}\sum_{i=1}^{n}\b_i|\w_i|^q
-\gamma_{str}\sum_{i=1}^{n}\w_i^2+\gamma_{str}\\
= \min_{\gamma_{str}\geq 0,\nu_{str}\geq 0}\max_{\w,\sum_{i=1}^n\b_i=2n-k,\b_i^2=1} \sum_{i=1}^{n}|\h_i||\w_i|
+\nu_{str}\sum_{i=1}^{n}\b_i|\w_i|^q
-\gamma_{str}\sum_{i=1}^{n}\w_i^2+\gamma_{str}\\
= \min_{\gamma_{str}\geq 0,\nu_{str}\geq 0} f_2(q,\h,\nu_{str},\gamma_{str},\b,\beta)+\gamma_{str}.\label{eq:streq1}
\end{multline}
where
\begin{equation}
f_2(q,\h,\nu_{str},\gamma_{str},\b,\beta)=\max_{\w,\sum_{i=1}^n\b_i=2n-k,\b_i^2=1} \left (\sum_{i=1}^{n}|\h_i||\w_i|
+\nu_{str}\sum_{i=1}^{n}\b_i|\w_i|^q
-\gamma_{str}\sum_{i=1}^{n}\w_i^2\right ).\label{eq:deff1str}
\end{equation}
One then has
\begin{multline}
w_D(\Sstr)=E\max_{\w\in\Sstr}\h^T\w=E\max_{\w\in\Sstr}f(\w)=
E\min_{\gamma_{str}\geq 0,\nu_{str}\geq 0} f_2(q,\h,\nu_{str},\gamma_{str},\b,\beta)+\gamma_{str}\\
\leq \min_{\gamma_{str}\geq 0,\nu_{str}\geq 0} E f_1(q,\h,\nu_{str},\gamma_{str},\b,\beta)+\gamma_{str}.\label{eq:wdineqstr}
\end{multline}
Now if one sets $\w_{i}=\frac{\w_{i}^{(s)}}{\sqrt{n}}$, $\gamma_{str}=\gamma_{str}^{(s)}\sqrt{n}$, and $\nu_{str}=\nu_{str}^{(s)}\sqrt{n}^{q-1}$ (where $\w_{i}^{(s)}$, $\gamma_{str}^{(s)}$, and $\nu_{str}^{(s)}$ are independent of $n$) then (\ref{eq:wdineqstr}) gives
\begin{multline}
\lim_{n\rightarrow\infty}\frac{w_D(\Sstr)}{\sqrt{n}}=\lim_{n\rightarrow\infty}\frac{E\max_{\w\in\Sstr}\h^T\w}{\sqrt{n}}
=\lim_{n\rightarrow\infty}\frac{E\max_{\w\in\Sstr}f(\w)}{\sqrt{n}}\\
\hspace{-.3in}=\lim_{n\rightarrow\infty}\frac{E\min_{\gamma_{str}\geq 0,\nu_{str}\geq 0} (f_2(q,\h,\nu_{str},\gamma_{str},\b,\beta)+\gamma_{str})}{\sqrt{n}}
\leq \lim_{n\rightarrow\infty}\frac{\min_{\gamma_{str}\geq 0,\nu_{str}\geq 0} (E f_2(q,\h,\nu_{str},\gamma_{str},\b,\beta)+\gamma_{str})}{\sqrt{n}}\\
=\min_{\gamma_{str}^{(s)}\geq 0,\nu_{str}^{(s)}\geq 0} (E f_2(q,\h,\nu_{str}^{(s)},\gamma_{str}^{(s)},\b,\beta)+\gamma_{str}^{(s)})
.\label{eq:wdineq1str}
\end{multline}
where using the machinery of \cite{StojnicCSetam09} one can assume that all quantities of interest concentrate and based on ideas of \cite{StojnicLiftStrSec13} (equation $(76)$) obtain
\begin{equation}
E f_2(q,\h,\nu_{str}^{(s)},\gamma_{str}^{(s)},\b,\beta)=\begin{cases}\max_{\w_i^{(s)}}(|\h_i||\w_i^{(s)}|+\nu_{str}^{(s)}|\w_i^{(s)}|^{q}
-\gamma_{str}^{(s)}(\w_i^{(s)})^2),
& |\h_i|\geq c_{\nu}\\
\max_{\w_i^{(s)}}(|\h_i||\w_i^{(s)}|-\nu_{str}^{(s)}|\w_i^{(s)}|^{q}
-\gamma_{str}^{(s)}(\w_i^{(s)})^2),
& |\h_i|\geq c_{\nu}\end{cases}.\label{eq:deff2str}
\end{equation}
As in \cite{StojnicLiftStrSec13}, one then finds $c_{\nu}$ from $\beta=\int_{|\h_i|\geq c_{\nu}}\frac{e^{-\frac{\h_i^2}{2}}d\h_i}{\sqrt{2\pi}}$. Clearly, $c_{\nu}=\sqrt{2}\mbox{erfinv}(1-\beta)$. For brevity we then write
\begin{equation}
E f_2(q,\h,\nu_{str}^{(s)},\gamma_{str}^{(s)},\b,\beta)=I_{str}^{(1)}+I_{str}^{(2)},\label{eq:deff21str}
\end{equation}
where
\begin{eqnarray}
I_{str}^{(1)} & = & E_{|\h_i|\geq c_{\nu}}\max_{\w_i^{(s)}}(|\h_i||\w_i^{(s)}|+\nu_{str}^{(s)}|\w_i^{(s)}|^q-\gamma_{str}^{(s)}(\w_i^{(s)})^2)\nonumber \\
I_{str}^{(2)} & = & E_{|\h_i|\leq c_{\nu}}\max_{\w_i^{(s)}}(|\h_i||\w_i^{(s)}|-\nu_{str}^{(s)}|\w_i^{(s)}|^q-\gamma_{str}^{(s)}(\w_i^{(s)})^2).\label{eq:defI1I2str}
\end{eqnarray}
We summarize the above results related to the strong threshold ($\beta_{str}^{(q)}$) in the following theorem.

\begin{theorem}(Strong threshold - lower bound)
Let $A$ be an $m\times n$ measurement matrix in (\ref{eq:system})
with i.i.d. standard normal components.  Let $\tilde{X}_{str}$ be the collection of all $k$-sparse vectors in $R^n$. Let $\tilde{\x}^{(i)}$ be any $k$-sparse vector from $\tilde{X}_{str}$. Further, assume that $\y^{(i)}=A\tilde{\x}^{(i)}$. Let $k,m,n$ be large
and let $\alpha=\frac{m}{n}$ and $\betastr^{(q)}=\frac{k}{n}$ be constants
independent of $m$ and $n$. Also set $c_{\nu}=\sqrt{2}\mbox{erfinv}(1-\beta_{str}^{(q)})$. Let
\begin{eqnarray}
I_{str}^{(1)} & = & E_{|\h_i|\geq c_{\nu}}\max_{\w_i^{(s)}}(|\h_i||\w_i^{(s)}|+\nu_{str}^{(s)}|\w_i^{(s)}|^q-\gamma_{str}^{(s)}(\w_i^{(s)})^2)\nonumber \\
I_{str}^{(2)} & = & E_{|\h_i|\leq c_{\nu}}\max_{\w_i^{(s)}}(|\h_i||\w_i^{(s)}|-\nu_{str}^{(s)}|\w_i^{(s)}|^q-\gamma_{str}^{(s)}(\w_i^{(s)})^2).\label{eq:defI1I2strthm}
\end{eqnarray}
If $\alpha$ and $\betastr^{(q)}$ are such that
\begin{equation}
\min_{\gamma_{str}^{(s)}\geq 0,\nu_{str}^{(s)}\geq 0} \left (\left (I_{str}^{(1)}
+I_{str}^{(2)}\right )+\gamma_{str}^{(s)}\right )<\sqrt{\alpha},\label{eq:strcondthmsec}
\end{equation}
then with overwhelming probability the solution of (\ref{eq:lq}) for every pair $(\y^{(i)},A)$ is the corresponding $k$-sparse $\tilde{\x}^{(i)}$.\label{thm:thmstrthrlq}
\end{theorem}
\begin{proof}
Follows from the above discussion.
\end{proof}

The results for the strong threshold obtained from the above theorem
are presented in Figure \ref{fig:str}. To be a bit more specific, as when we presented the corresponding results for the sectional thresholds in the previous section we selected four different values of $q$, namely $q\in\{0,0.1,0.3,0.5\}$ in addition to standard $q=1$ case already discussed in \cite{StojnicCSetam09}.
\begin{figure}[htb]
\centering
\centerline{\epsfig{figure=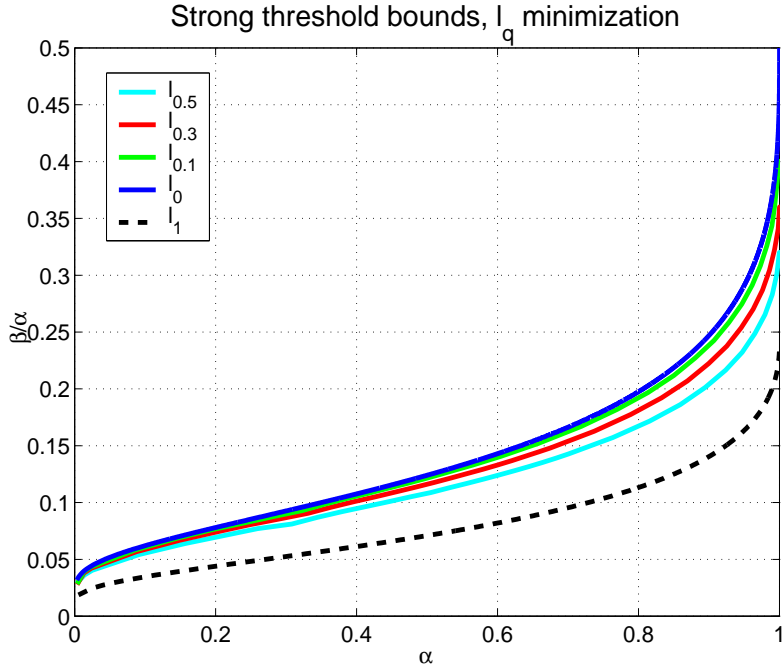,width=10.5cm,height=9cm}}
\caption{\emph{Strong} threshold, $\ell_q$-optimization}
\label{fig:str}
\end{figure}
As can be seen from Figure \ref{fig:sec}, the results are better than for $q=1$. Moreover, they hint that as $q$ is decreasing the strong thresholds are increasing a fact one may naturally expect. Of course, it is fairly obvious (as was when we studied sectional thresholds) our methodology works successfully to a degree, i.e. its a lower-bounding tendency eventually comes into a full effect. While to see that when for example $q=1$ one needs a quite extra knowledge (see, e.g. \cite{DonohoPol,StojnicLiftStrSec13}) it is quite obvious when $q=0$. In that case the true threshold should be substantially higher.

The obtained results can also be compared with the best known ones for $\ell_1$-minimization from \cite{StojnicLiftStrSec13,DonohoPol} as well. These are slightly above the curve that corresponds to $\ell_1$ given in Figure \ref{fig:sec}; however, since these use a more sophisticated methodology we skip adding them and making the figure even more detailed.

Also, as in the previous section, we again emphasize that all results are obtained after numerical computations (all of those were done in pretty much the same fashion as explained in the previous section). Since solving the maximization over $\w$
even on a global optimum level may again jeopardize the lower-bounding rigorousness of the presented results in addition to plots in Figure \ref{fig:str} we present in Tables \ref{tab:strtab1}, \ref{tab:strtab2}, and \ref{tab:strtab3} the concrete values we obtained for $\nu_{str}^{(s)}$ and
$\gamma_{str}^{(s)}$ for certain $\beta_{str}^{(q)}$ on the way to computing corresponding $\alpha$ (as indicated above the tables, Table \ref{tab:strtab1} contains data for $\ell_q, q=0.5$, Table \ref{tab:strtab2} contains data for $\ell_q, q=0.3$, and Table \ref{tab:strtab3} contains data for $\ell_q, q=0.1$). That way the
interested reader can again double check if the optimization over $\w$ in any way endangered the lower-bounding rigorousness.
Of course, as in the previous section, we again do reemphasize that the results presented in the above theorem are completely rigorous,
it is just that some of the numerical work that we performed could have been a bit imprecise
(we again firmly believe that this is not the case).

\begin{table}
\caption{Strong threshold bounds $\ell_q,q=0.5$}\vspace{0in}
\centering\hspace{-.3in}
\begin{tabular}{||c|c|c|c|c|c|c|c|c|c|c|c||}\hline\hline
$\beta_{str}^{(q)}$  & $0.0005$ &  $0.0050$ & $0.0150$ & $0.0250$ & $0.0400$ & $0.0550$ & $0.0750$ & $0.1000$ & $0.1400$ & $0.1800$ & $0.3200$  \\ \hline\hline
$\alpha$             & $0.0138$ &  $0.0919$ & $0.2114$ & $0.3081$ & $0.4142$ & $0.5053$ & $0.6030$ & $0.7006$ & $0.8156$ & $0.8944$ & $0.9998$  \\ \hline
$\nu_{str}^{(s)}$    & $9.2604$ &  $3.8721$ & $2.5000$ & $2.1680$ & $1.4450$ & $1.2500$ & $0.9423$ & $0.7368$ & $0.5141$ & $0.3577$ & $0.0286$  \\ \hline
$\gamma_{str}^{(s)}$ & $0.0587$ &  $0.1563$ & $0.2267$ & $0.2612$ & $0.3217$ & $0.3499$ & $0.3881$ & $0.4183$ & $0.4514$ & $0.4727$ & $0.4996$  \\ \hline\hline
\end{tabular}
\label{tab:strtab1}
\end{table}

\begin{table}
\caption{Strong threshold bounds $\ell_q,q=0.3$}\vspace{0in}
\hspace{-0in}\centering
\begin{tabular}{||c|c|c|c|c|c|c|c|c|c|c|c||}\hline\hline
$\beta_{str}^{(q)}$  & $0.0005$ & $0.0050$ & $0.0150$ & $0.0250$ & $0.0400$ & $0.0600$ & $0.0800$ & $0.1000$ & $0.1400$ & $0.2000$ & $0.3600$  \\ \hline\hline
$\alpha$             & $0.0132$ & $0.0879$ & $0.2020$ & $0.2918$ & $0.3968$ & $0.5100$ & $0.6007$ & $0.6752$ & $0.7888$ & $0.8995$ & $0.9999$  \\ \hline
$\nu_{str}^{(s)}$    & $17.763$ & $5.6990$ & $3.2832$ & $2.6563$ & $1.7745$ & $1.3136$ & $1.0333$ & $0.8362$ & $0.5737$ & $0.3330$ & $0.0259$\\ \hline
$\gamma_{str}^{(s)}$ & $0.0568$ & $0.1563$ & $0.2245$ & $0.2582$ & $0.3147$ & $0.3567$ & $0.3872$ & $0.4104$ & $0.4436$ & $0.4737$ & $0.4994$ \\ \hline\hline
\end{tabular}
\label{tab:strtab2}
\end{table}

\begin{table}
\caption{Strong threshold bounds $\ell_q,q=0.1$}\vspace{0in}
\hspace{-0in}\centering
\begin{tabular}{||c|c|c|c|c|c|c|c|c|c|c|c||}\hline\hline
$\beta_{str}^{(q)}$  & $0.0005$ & $0.0050$ & $0.0150$ & $0.0250$ & $0.0400$ & $0.0600$ & $0.0850$ & $0.1200$ & $0.1600$ & $0.2200$ & $0.4000$   \\ \hline\hline
$\alpha$             & $0.0128$ & $0.0858$ & $0.1966$ & $0.2843$ & $0.3862$ & $0.4963$ & $0.6045$ & $0.7187$ & $0.8132$ & $0.9070$ & $0.9991$   \\ \hline
$\nu_{str}^{(s)}$    & $34.531$ & $8.5931$ & $4.5230$ & $3.5547$ & $2.1967$ & $1.5735$ & $1.1320$ & $0.7736$ & $0.5276$ & $0.3035$ & $0.0223$   \\ \hline
$\gamma_{str}^{(s)}$ & $0.0562$ & $0.1563$ & $0.2215$ & $0.2518$ & $0.3125$ & $0.3519$ & $0.3883$ & $0.4234$ & $0.4504$ & $0.4756$ & $0.4993$   \\ \hline\hline
\end{tabular}
\label{tab:strtab3}
\end{table}

\subsection{Special cases}
\label{sec:strthrspecial}

As when we studied the sectional thresholds, in this subsection we briefly note that some of the above computations can be done in a faster, more explicit fashion.

\subsubsection{$q\rightarrow 0$}
\label{sec:strthrspecialq0}

The first case we consider is $q=0$. The curve for that case given in Figure \ref{fig:str} can be obtained in a more direct fashion without all the computations required by Theorem \ref{thm:thmstrthrlq}. Here is a brief sketch how one can proceed. Let
\begin{equation}
\tilde{\h}=[\h_{(1)},\h_{(2)},\dots,\h_{(n)}],\label{eq:defhq0str}
\end{equation}
where
$[\h_{(1)},\h_{(2)},\dots,\h_{(n)}]$ are the absolute values of components of $[\h_{1},\h_{2},\dots,\h_{n}]$ sorted in an increasing order. Then one has
\begin{multline}
\lim_{n\rightarrow\infty}\frac{w_D(\Sstr)}{\sqrt{n}}=\lim_{n\rightarrow\infty}\frac{E\max_{\w\in\Sstr}\h^T\w}{\sqrt{n}}
=\lim_{n\rightarrow\infty}\frac{E\max_{\w\in\Sstr}\sum_{i=1}^{n}\tilde{\h}_i|\w_i|}{\sqrt{n}}\\
=\lim_{n\rightarrow\infty}\frac{E\sqrt{\sum_{i=n-2k+1}^{n}\h_{(i)}^2}}{\sqrt{n}}
\leq \lim_{n\rightarrow\infty}\frac{\sqrt{E\sum_{i=n-2k+1}^{n}\h_{(i)}^2}}{\sqrt{n}}\label{eq:wdlq0str}
\end{multline}
Applying the machinery of \cite{StojnicCSetam09} then gives
\begin{multline}
\lim_{n\rightarrow\infty}\frac{w_D(\Sstr)}{\sqrt{n}}
\leq \lim_{n\rightarrow\infty}\frac{\sqrt{E\sum_{i=n-2k+1}^{n}\h_{(i)}^2}}{\sqrt{n}}\\
=\sqrt{2\beta_{str}^{(0)}+\frac{2}{\sqrt{\pi}}\mbox{erfinv}(1-2\beta_{str}^{(0)})
e^{-\left (\mbox{erfinv}(1-2\beta_{str}^{(0)})\right )^2}}.\label{eq:wdlq0}
\end{multline}
Equalling the quantity on the right hand side with $\sqrt{\alpha}$ then gives the characterization of $\ell_0$ curve in Figure \ref{fig:str}.

\subsubsection{$q=\frac{1}{2}$}
\label{sec:strthrspecialq05}

As when we studied the sectional threshold in the previous section, another special case that allows a further simplification of the results presented in Theorem \ref{thm:thmsecthrlq} is when $q=\frac{1}{2}$. In that case one can apply the strategy that led to (\ref{eq:cubicq05}) to obtain its a strong threshold analogue
\begin{eqnarray}
& & |\h_i|\pm\frac{1}{2}\nu_{str}^{(s)}|\w_i^{(s)}|^{-1/2}-2\gamma_{str}^{(s)}|\w_i^{(s)}|=0\nonumber \\
& \Leftrightarrow & |\h_i|\sqrt{|\w_i^{(s)}|}\pm\frac{1}{2}\nu_{str}^{(s)}-2\gamma_{str}^{(s)}\sqrt{|\w_i^{(s)}|}^{3}=0.\label{eq:cubicq05str}
\end{eqnarray}
This is a cubic equation and can be solved explicitly which of course substantially facilitates the integrations over $\h_i$. Also, as mentioned earlier, similar strategy can be adopted for other rational $q$ but the ``explicit" solutions soon become more complicated than the numerical ones and we skip presenting them.

\section{$\ell_q$-minimization weak threshold}
\label{sec:weakthr}

In this section we assess the performance of $\ell_q$ minimization by looking at its weak thresholds. Before proceeding further, as in the previous section, we slightly readjust the definition of the $\ell_1$ weak thresholds given earlier in the $\ell_1$ minimization context so that it fits the $\ell_q$ case considered here. Namely, one considers a scenario where for any given constant
$\alpha\leq 1$ and \emph{a given fixed} $\tilde{\x}$ in \ref{eq:yrepsystem}
there will be a maximum allowable value of $\beta$ such that the solution of
(\ref{eq:lq}) is that given $\tilde{\x}$ with overwhelming
probability. We will refer to such a $\beta$ as the \emph{weak threshold} and will denote it by $\beta_{weak}^{(q)}$ (we again recall that more on similar definitions of the weak threshold the interested reader can find in e.g. \cite{DonohoPol,StojnicCSetam09}).

\subsection{Weak threshold preliminaries}
\label{sec:weakthrprelim}

Below we will provide a way to quantify behavior of $\beta_{weak}^{(q)}$. As usual we rely on some of the mechanisms presented in \cite{StojnicCSetam09} and some of those presented in Section \ref{sec:secthr}. Along the same lines, we will continue to assume a substantial level of familiarity with many of the well-known results that relate to the performance characterization of (\ref{eq:l1})  and will fairly often recall on many results/definitions that we established in \cite{StojnicCSetam09}. We start by introducing a nice way of characterizing weak success/failure of (\ref{eq:lq}).

\begin{theorem}(A given fixed $\x$)
Assume that an $m\times n$ matrix $A$ is given. Let $\tilde{\x}$ be a $k$-sparse vector and let $\tilde{\x}_1=\tilde{\x}_2=\dots=\tilde{\x}_{n-k}=0$. Further, assume that $\y=A\tilde{\x}$ and that $\w$ is
an $n\times 1$ vector. If
\begin{equation}
(\forall \w\in \textbf{R}^n | A\w=0) \quad  \sum_{i=1}^{n-k}|\w_i|^q+\sum_{i=n-k+1}^n|\tilde{\x}_i+\w_i|^q>\sum_{i=n-k+1}^{n}|\tilde{\x}_{i}|^q
\label{eq:thmeqgenweak1}
\end{equation}
then the solution of (\ref{eq:lq}) obtained for pair $(\y,A)$ is $\tilde{\x}$.
\label{thm:thmgenweak}
\end{theorem}
\begin{proof}
The proof follows directly from the corresponding results for $\ell_1$ (see, e.g. Theorem $2$ in \cite{StojnicICASSP09} and references therein). For the completeness we just sketch the argument again. Let $\hatx$ be the solution of (\ref{eq:lq}). We want to show that if (\ref{eq:thmeqgenweak1}) holds then $\hatx=\tilde{\x}$. To that end assume opposite, i.e. assume that (\ref{eq:thmeqgenweak1}) holds but $\hatx\neq\tilde{\x}$. Then since $\y=A\hatx$ and $\y=A\tilde{\x}$ one must have $\hatx =\tilde{\x}+\w$ with $\w$ such that $A\w=0$. Also, since $\hatx$ is the solution of (\ref{eq:lq}) one has that
\begin{equation}
\sum_{i=1}^{n-k}|\w_i|^q+\sum_{i=n-k+1}^n|\tilde{\x}_i+\w_i|^q=\sum_{i=1}^n|\tilde{\x}_i+\w_i|^q\leq \sum_{i=1}^{n}|\tilde{\x}_i|^q=\sum_{i=n-k+1}^{n}|\tilde{\x}_i|^q.\label{eq:absvalweak}
\end{equation}
Clearly, (\ref{eq:absvalweak}) contradicts (\ref{eq:thmeqgenweak1}) and $\hatx\neq\tilde{\x}$ can not hold. Therefore $\hatx=\tilde{\x}$ which is exactly what the theorem claims.
\end{proof}

\noindent \textbf{Remark:} As earlier, the above proof is nothing original. It simply follows the well known arguments for $\ell_1$ case.

We then, following the methodology of \cite{StojnicCSetam09},
start by defining a set $\Sweak$
\begin{equation}
\Sweak(\tilde{\x})=\{\w\in S^{n-1}| \quad \sum_{i=n-k+1}^n |\tilde{\x}_i|^q\geq \sum_{i=1}^{n-k}|\w_i|^q+\sum_{i=n-k+1}^n|\tilde{\x}_i+\w_i|^q\},\label{eq:defSweak}
\end{equation}
where $S^{n-1}$ is the unit sphere in $R^n$. To continue following methodology of \cite{StojnicCSetam09} we will utilize the following slight modification of Theorem \ref{thm:Gordonmesh}.
\begin{theorem}(\cite{Gordon88} Escape through a mesh)
\label{thm:Gordonmesh} Let $S(\x)$ be a collection of subsets of the unit Euclidean
sphere $S^{n-1}$ in $R^{n}$ indexed by a collection of vectors $\x$. Let $Y$ be a random
$(n-m)$-dimensional subspace of $R^{n}$, spanned by $(n-m)$ vectors from $R^n$ with i.i.d. standard normal components. Let
\begin{eqnarray}
w_D(S(\x)) & = & E\sup_{\w\in S} (\h^T\w)\nonumber \\
\max_{\x}w_D(S(\x)) & = & \max_{\x}E\sup_{\w\in S} (\h^T\w) \label{eq:widthdefweak}
\end{eqnarray}
where $\h$ is a random column vector in $R^{n}$ with i.i.d. standard normal components. Assume that
$\max_{\x}w_D(S(\x))<\left ( \sqrt{m}-\frac{1}{4\sqrt{m}}\right )$. Select a subset of $S(\x)$, say $S(\x^{(i)})$. Then
\begin{equation}
P(Y\cap S(\x^{(i)})=0)>1-3.5e^{-\frac{\left (
\sqrt{m}-\frac{1}{4\sqrt{m}}-\max_{\x}(w_D(S(\x^{(i)}))) \right ) ^2}{18}}.
\label{eq:thmesh}
\end{equation}
\end{theorem}
\begin{proof}
It is a trivial extension of the Gordon's original proof of Theorem \ref{thm:Gordonmesh}.
\end{proof}

The methodology of \cite{StojnicCSetam09} then proceeds by characterizing
\begin{equation}
\max_{\tilde{\x}}w_D(\Sweak(\tilde{\x}))=\max_{\tilde{\x}}E\max_{\w\in \Sweak(\tilde{\x})} (\h^T\w),\label{eq:negham1}
\end{equation}
where to facilitate the exposition we, as earlier, replace $\sup$ with a $\max$. Below we present a way to create an upper-bound on $w_D(\Sweak(\tilde{\x}))$. Equalling such an upper bound with $\sqrt{m}$ would be roughly enough to provide a characterization of the weak thresholds.

\subsection{Weak threshold computation}
\label{sec:weakthrcomp}

We recall that $f(\w)=\h^T\w$ and
we start with the following line of identities
\begin{multline}
\hspace{-.5in}\max_{\w\in\Sweak(\tilde{\x})}f(\w)=-\min_{\w\in\Sweak(\tilde{\x})} -\h^T\w\\
\hspace{-.7in}=-\min_{\w}\max_{\gamma_{weak}\geq 0,\nu_{weak}\geq 0} -\h^T\w+
\nu_{weak}\sum_{i=n-k+1}^{n}|\tilde{\x}_i+\w_i|^q
+\nu_{weak}\sum_{i=1}^{n-k}|\w_i|^q-\nu_{weak}\sum_{i=n-k+1}^{n}|\tilde{\x}_i|^q+\gamma_{weak}\sum_{i=1}^{n}\w_i^2-\gamma_{weak}\\
\hspace{-.5in}\leq -\max_{\gamma_{weak}\geq 0,\nu_{weak}\geq 0}\min_{\w} -\h^T\w
+\nu_{weak}\sum_{i=n-k+1}^{n}|\x_i+\w_i|^q
+\nu_{weak}\sum_{i=1}^{n-k}|\w_i|^q-\nu_{weak}\sum_{i=n-k+1}^{n}|\x_i|^q+\gamma_{weak}\sum_{i=1}^{n}\w_i^2-\gamma_{weak}\\
\hspace{-.5in}= \min_{\gamma_{weak}\geq 0,\nu_{weak}\geq 0}\max_{\w} \h^T\w
-\nu_{weak}\sum_{i=n-k+1}^{n}|\x_i+\w_i|^q
-\nu_{weak}\sum_{i=1}^{n-k}|\w_i|^q+\nu_{weak}\sum_{i=n-k+1}^{n}|\x_i|^q-\gamma_{weak}\sum_{i=1}^{n}\w_i^2+\gamma_{weak}\\
\hspace{-.7in}=\min_{\gamma_{weak}\geq 0,\nu_{weak}\geq 0}\max_{\w} \sum_{i=n-k+1}^{n}(\h_i\w_i-\nu_{weak}|\tilde{\x}_i+\w_i|^q+\nu_{weak}|\tilde{\x}_i|^q-\gamma_{weak}\w_i^2)
+\sum_{i=1}^{n-k}(\h_i|\w_i|-\nu_{weak}|\w_i|^q-\gamma_{weak}\w_i^2)+\gamma_{weak}\\
=\min_{\gamma_{weak}\geq 0,\nu_{weak}\geq 0} f_3(q,\h,\nu_{weak},\gamma_{weak},\beta)+\gamma_{weak},\label{eq:weakeq1}
\end{multline}
where
\begin{multline}
\hspace{0in}f_3(q,\h,\nu_{weak},\gamma_{weak},\tilde{\x},\beta)=\max_{\w} (\sum_{i=n-k+1}^{n}(\h_i\w_i-\nu_{weak}|\tilde{\x}_i+\w_i|^q+\nu_{weak}|\tilde{\x}_i|^q-\gamma_{weak}\w_i^2)\\
+\sum_{i=1}^{n-k}(\h_i|\w_i|-\nu_{weak}|\w_i|^q-\gamma_{weak}\w_i^2)).\label{eq:deff1weak}
\end{multline}
One then has
\begin{multline}
\hspace{0in}\max_{\tilde{\x}}w_D(\Sweak(\tilde{\x}))=\max_{\tilde{\x}}E\max_{\w\in\Sweak(\tilde{\x})}\h^T\w
=\max_{\tilde{\x}}E\max_{\w\in\Sweak(\tilde{\x})}f(\w)=\\
\hspace{-.5in}\max_{\tilde{\x}}E\min_{\gamma_{weak}\geq 0,\nu_{weak}\geq 0} f_3(q,\h,\nu_{weak},\gamma_{weak},\tilde{\x},\beta)+\gamma_{weak}
\leq \max_{\tilde{\x}}\min_{\gamma_{weak}\geq 0,\nu_{weak}\geq 0} E f_3(q,\h,\nu_{weak},\gamma_{weak},\tilde{\x},\beta)+\gamma_{weak}.\label{eq:wdineqweak}
\end{multline}
Now if one sets $\w_{i}=\frac{\w_{i}^{(s)}}{\sqrt{n}}$, $\gamma_{weak}=\gamma_{weak}^{(s)}\sqrt{n}$, and $\nu_{weak}=\nu_{weak}^{(s)}\sqrt{n}^{q-1}$ (where $\w_{i}^{(s)}$, $\gamma_{weak}^{(s)}$, and $\nu_{weak}^{(s)}$ are independent of $n$) then (\ref{eq:wdineq}) gives
\begin{multline}
\lim_{n\rightarrow\infty}\frac{\max_{\tilde{\x}}w_D(\Sweak(\tilde{\x}))}{\sqrt{n}}=\lim_{n\rightarrow\infty}
\frac{\max_{\tilde{\x}}E\max_{\w\in\Sweak(\tilde{\x})}\h^T\w}{\sqrt{n}}
=\lim_{n\rightarrow\infty}\frac{\max_{\tilde{\x}}E\max_{\w\in\Sweak(\tilde{\x})}f(\w)}{\sqrt{n}}\\
\hspace{-.3in}=\lim_{n\rightarrow\infty}\frac{\max_{\tilde{\x}}E\min_{\gamma_{weak}\geq 0,\nu_{weak}\geq 0} (f_3(q,\h,\nu_{weak},\gamma_{weak},\tilde{\x},\beta)+\gamma_{weak})}{\sqrt{n}}\\
\leq \lim_{n\rightarrow\infty}\frac{\max_{\tilde{\x}}\min_{\gamma_{weak}\geq 0,\nu_{weak}\geq 0} (E f_3(q,\h,\nu_{weak},\gamma_{weak},\tilde{\x},\beta)+\gamma_{weak})}{\sqrt{n}}\\
\hspace{-.5in}=\max_{\tilde{\x}_i,i>n-k}\min_{\gamma_{weak}^{(s)}\geq 0,\nu_{weak}^{(s)}\geq 0}  ((\beta E\max_{\w_i^{(s)}}(\h_i\w_i^{(s)}+\nu_{weak}^{(s)}|\tilde{\x}_i+\w_i^{(s)}|^q-\nu_{weak}^{(s)}|\tilde{\x}_i|^q-\gamma_{weak}^{(s)}(\w_i^{(s)})^2)\\
+(1-\beta)E\max_{\w_j^{(s)}}(|\h_j||\w_j^{(s)}|-\nu_{weak}^{(s)}|\w_j^{(s)}|^q-\gamma_{weak}^{(s)}(\w_j^{(s)})^2) )+\gamma_{weak}^{(s)} )\\
=\max_{\tilde{\x}_i,i>n-k}\min_{\gamma_{weak}^{(s)}\geq 0,\nu_{weak}^{(s)}\geq 0} \left (\left (\beta I_{weak}^{(1)}
+(1-\beta)I_{weak}^{(2)}\right )+\gamma_{weak}^{(s)}\right ),\label{eq:wdineq1weak}
\end{multline}
where
\begin{eqnarray}
I_{weak}^{(1)} & = & E\max_{\w_i^{(s)}}(\h_i\w_i^{(s)}+\nu_{weak}^{(s)}|\tilde{\x}_i+\w_i^{(s)}|^q-\nu_{weak}^{(s)}|\tilde{\x}_i|^q-\gamma_{weak}^{(s)}(\w_i^{(s)})^2)\nonumber \\
I_{weak}^{(2)} & = & E\max_{\w_j^{(s)}}(|\h_j||\w_j^{(s)}|-\nu_{weak}^{(s)}|\w_j^{(s)}|^q-\gamma_{weak}^{(s)}(\w_j^{(s)})^2).\label{eq:defI1I2weak}
\end{eqnarray}
We summarize the above results related to the weak threshold ($\beta_{weak}^{(q)}$) in the following theorem.

\begin{theorem}(Weak threshold - lower bound)
Let $A$ be an $m\times n$ measurement matrix in (\ref{eq:system})
with i.i.d. standard normal components. Let $\tilde{\x}\in R^n$ be a $k$-sparse vector for which $\tilde{\x}_1=0,\tilde{\x}_2=0,,\dots,\tilde{\x}_{n-k}=0$ and let  $\y=A\tilde{\x}$. Let $k,m,n$ be large
and let $\alpha=\frac{m}{n}$ and $\betaweak^{(q)}=\frac{k}{n}$ be constants
independent of $m$ and $n$. Let
\begin{eqnarray}
I_{weak}^{(1)} & = & E\max_{\w_i^{(s)}}(\h_i\w_i^{(s)}+\nu_{weak}^{(s)}|\tilde{\x}_i+\w_i^{(s)}|^q-\nu_{weak}^{(s)}|\tilde{\x}_i|^q-\gamma_{weak}^{(s)}(\w_i^{(s)})^2)\nonumber \\
I_{weak}^{(2)} & = & E\max_{\w_j^{(s)}}(|\h_j||\w_j^{(s)}|-\nu_{weak}^{(s)}|\w_j^{(s)}|^q-\gamma_{weak}^{(s)}(\w_j^{(s)})^2).\label{eq:defI1I2weakthm}
\end{eqnarray}
If $\alpha$ and $\betaweak^{(q)}$ are such that
\begin{equation}
\max_{\tilde{\x}_i,i>n-k}\min_{\gamma_{weak}^{(s)}\geq 0,\nu_{weak}^{(s)}\geq 0} \left (\left (\beta_{weak}^{(q)} I_{weak}^{(1)}
+(1-\beta_{weak}^{(q)})I_{weak}^{(2)}\right )+\gamma_{weak}^{(s)}\right )<\sqrt{\alpha},\label{eq:weakcondthmweak}
\end{equation}
then with overwhelming probability the solution of (\ref{eq:lq}) obtained for pair $(\y,A)$ is $\tilde{\x}$.\label{thm:thmweakthrlq}
\end{theorem}
\begin{proof}
Follows from the above discussion.
\end{proof}

The results for the weak threshold obtained from the above theorem
are presented in Figure \ref{fig:weak}. To be a bit more specific, we selected three different values of $q$, namely $q\in\{0,0.3,0.5\}$ in addition to standard $q=1$ case already discussed in \cite{StojnicCSetam09} (we skipped the $q=0.1$ case that we considered in earlier sections since now one has an extra optimization to perform and when $q$ is small additional precision/computaion time may be needed to obtain valid results).
\begin{figure}[htb]
\centering
\centerline{\epsfig{figure=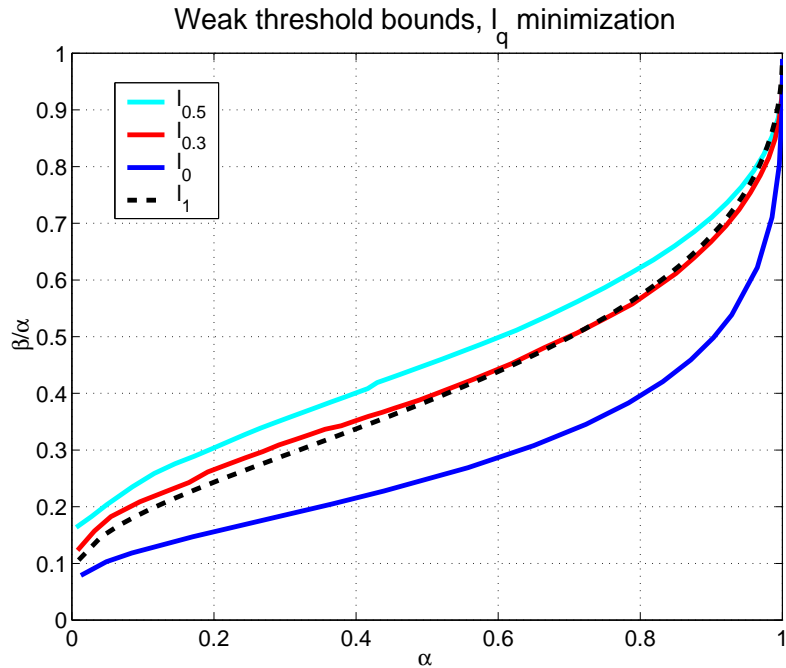,width=10.5cm,height=9cm}}
\caption{\emph{Weak} threshold, $\ell_q$-optimization}
\label{fig:weak}
\end{figure}
As can be seen from Figure \ref{fig:weak}, for some values of $q$ the results are better than for $q=1$. However, as was the case when we studied the sectional thresholds, for some $q$'s and some ranges of $\alpha$ the results are worse. Of course one again has to be careful how to interpret this. As was the case when we studied the sectional thresholds, one may naturally expect that as $q$ goes down the threshold results become better, i.e. the resulting curves go up. That again does happen down to some values for $q$; however, after that the curves start sliding down and eventually for $q=0$ we actually have a curve that is even below $q=1$ case. Of course this again just shows that our methodology works successfully to a degree, i.e. its a lower-bounding tendency eventually comes into a full effect. Of course, as earlier, if one is interested in the best possible weak threshold values for any $q$ rather than the methodology itself the curves that go down as $q$ goes up could be ignored. However, we again kept them on the plot to emphasize that the proposed methodology has some inherent deficiencies.

Also, as almost all other results we presented so far, the results we presented in Figure \ref{fig:weak} are obtained after numerical computations. They mostly included numerical optimizations which were all (except maximization over $\w$) done on a local optimum level. We do not know how (if in any way) solving them on a global optimum level would affect the location of the plotted curves. Also, as earlier, numerical integrations were done on a finite precision level as well which could have potentially harmed the final results as well. Still, we believe that the methodology can not achieve substantially more than what we presented in Figure \ref{fig:weak} (and hopefully is not severely degraded with numerical integrations and maximization over $\w$ and $\tilde{\x}_i$).

It is important to emphasize that as in the case when we studied the sectional thresholds in Section \ref{sec:secthr}, solving over $\nu_{weak}^{(s)}$ and $\gamma_{weak}^{(s)}$ on a local optimum level may lower the curves but it
certainly does not jeopardize their lower bounding rigorousness. However, solving the maximization over $\w$ even on a global optimum level as we did, may do so. Moreover, one now also has to solve maximization over $\tilde{\x}_i$ on a global optimum level. We have not done so and it is possible that such an imprecision made curves be higher than they really are. Since this may jeopardize the lower bounding rigorousness
in addition to plots in Figure \ref{fig:weak} we again present in Tables \ref{tab:weaktab1} and \ref{tab:weaktab2} the concrete values we obtained for $\nu_{weak}^{(s)}$,
$\gamma_{weak}^{(s)}$, and $\tilde{\x}_i$ for certain $\beta_{weak}^{(q)}$ on the way to computing corresponding $\alpha$ (as indicated above the tables, Table \ref{tab:weaktab1} contains data for $\ell_q, q=0.5$ and Table \ref{tab:weaktab2} contains data for $\ell_q, q=0.3$). That way the
interested reader can double check if the optimization over $\w$ in any way endangered the lower-bounding rigorousness.
Of course, as mentioned on a couple of occasions earlier, we do reemphasize that the results presented in the above theorem are completely rigorous,
it is just that some of the numerical work that we performed could have been a bit imprecise. Also, as earlier, we firmly believe that all the numerical work with the exception of optimization over $\tilde{\x}_i$ did not make any substantial imprecisions. When it comes to optimization over $\tilde{\x}_i$, such an optimization is not that hard to implement (if needed) even as a variant of the exhaustive search. However, solving it numerically would require a bit more computational time and we opted for potentially suboptimal local search. We do emphasize though, that with a bit more time available it should not be that much of a problem to double check if our potential sub-optimality in any way endangered the rigorousness of the presented plots. We believe that it is not case but have not done a complete exhaustive search to confirm such a belief. Although it is not much of a guarantee for anything, we do mention that the curve we obtained for $q=0$ closely matches the one that can be obtained when performing the exact optimizations and integrations (when $q=0$ these are a bit involved but as mentioned below possible). In other words, apart from standard finite precision problems one unavoidably has the blue curve in Figure \ref{fig:weak} is roughly speaking where it really should be. Of course, that is not of much use since this curve is anyway below the $\ell_1$. However, as we just mentioned, it may be used as an indication that even when it comes to $q=0.5$ and $q=0.3$ maybe our numerical work is not that much off (if at all).

\subsection{Special cases}
\label{sec:weakthrspecial}

One can again create a substantial simplification of results given in Theorem \ref{thm:thmgenweak} for certain values of $q$. For example, for $q=0$ or $q=1/2$ one can follow the strategy of previous sections and simplify some of the computations. However, such results (while simpler than those from Theorem \ref{thm:thmgenweak}) are still not very simple. Moreover, since for $q=0$ the results one eventually obtains are not even better than the well known ones for $\ell_1$ we skip presenting these simplifications.

\begin{table}
\caption{Weak threshold bounds $\ell_q,q=0.5$}\vspace{0in}
\centering\hspace{-.3in}
\begin{tabular}{||c|c|c|c|c|c|c|c|c|c|c|c||}\hline\hline
$\beta_{weak}^{(q)}$  & $0.0050$ &  $0.0200$ & $0.0600$ & $0.1100$ & $0.1600$ & $0.2400$ & $0.3200$ & $0.4000$ & $0.5200$ & $0.6400$ & $0.9200$  \\ \hline\hline
$\alpha$              & $0.0274$ &  $0.0851$ & $0.1981$ & $0.3071$ & $0.3995$ & $0.5212$ & $0.6257$ & $0.7117$ & $0.8185$ & $0.9006$ & $0.9990$  \\ \hline
$\tilde{\x}_i$        & $7.4176$ &  $4.5521$ & $2.5595$ & $2.0168$ & $1.7742$ & $1.3513$ & $1.2865$ & $1.2250$ & $1.2925$ & $1.3074$ & $1.6199$  \\ \hline
$\nu_{weak}^{(s)}$    & $6.7123$ &  $3.8539$ & $2.4321$ & $1.7927$ & $1.4733$ & $1.1500$ & $0.9033$ & $0.7319$ & $0.5583$ & $0.3900$ & $0.0442$  \\ \hline
$\gamma_{weak}^{(s)}$ & $0.0830$ &  $0.1505$ & $0.2212$ & $0.2812$ & $0.3178$ & $0.3470$ & $0.3931$ & $0.4206$ & $0.4535$ & $0.4777$ & $0.5018$  \\ \hline\hline
\end{tabular}
\label{tab:weaktab1}
\end{table}

\begin{table}
\caption{Weak threshold bounds $\ell_q,q=0.3$}\vspace{0in}
\centering\hspace{-.3in}
\begin{tabular}{||c|c|c|c|c|c|c|c|c|c|c|c||}\hline\hline
$\beta_{weak}^{(q)}$  & $0.0010$ &  $0.0200$ & $0.0500$ & $0.0900$ & $0.1400$ & $0.2000$ & $0.2800$ & $0.3600$ & $0.4400$ & $0.6000$ & $0.9200$  \\ \hline\hline
$\alpha$              & $0.0081$ &  $0.0958$ & $0.1913$ & $0.2914$ & $0.3985$ & $0.5054$ & $0.6188$ & $0.7110$ & $0.7889$ & $0.8993$ & $0.9991$  \\ \hline
$\tilde{\x}_i$        & $9.7741$ &  $2.9006$ & $1.6855$ & $1.5349$ & $0.8705$ & $0.8734$ & $0.8656$ & $0.9196$ & $0.8888$ & $0.9157$ & $1.4012$  \\ \hline
$\nu_{weak}^{(s)}$    & $22.565$ &  $5.4895$ & $3.3171$ & $2.3060$ & $1.7590$ & $1.3475$ & $0.9736$ & $0.7632$ & $0.5694$ & $0.3784$ & $0.0368$  \\ \hline
$\gamma_{weak}^{(s)}$ & $0.0442$ &  $0.1519$ & $0.2213$ & $0.2785$ & $0.3080$ & $0.3436$ & $0.3885$ & $0.4258$ & $0.4436$ & $0.4737$ & $0.5006$  \\ \hline\hline
\end{tabular}
\label{tab:weaktab2}
\end{table}

\section{Conclusion}
\label{sec:conc}

In this paper we looked at classical under-determined linear systems with sparse solutions. We analyzed a particular optimization technique called $\ell_q$ optimization. While its a convex counterpart $\ell_1$ technique is known to work well often it is a much harder task to determine if
$\ell_q$ exhibits a similar or better behavior; and especially if it exhibits a better behavior how much better quantitatively it is. We made some sort of progress in this direction in this paper. Namely, we showed that in many cases the $\ell_q$ would provide stronger guarantees than $\ell_1$ and in many other ones we provided bounds that are better than the ones we could provide for $\ell_1$. Of course, having better bounds does not guarantee that the performance is better as well but in our view serves as a solid indication that overall, $\ell_q,q<1$, should work better than $\ell_1$.

To be a bit more specific, in this paper we looked at sectional, strong, and weak thresholds of the $\ell_q$ optimization. We created a mechanism that can help provide lower bounds on all of these thresholds. The methodology is especially valuable since the underlying problems are non-convex and some of them actually highly combinatorial. That makes them incredibly hard to analyze. However, using the powerful methodology we recently developed \cite{StojnicCSetam09} we were able to attack all these problems and make a substantial progress in characterizing the thresholds they eventually produce.

Of course, much more can be done, including generalizations of the presented concepts to many other variants of these problems. The examples include various different unknown vector structures (a priori known to be positive vectors, block-sparse, binary/box constrained vectors etc.), various noisy versions (approximately sparse vectors, noisy measurements $\y$), low rank matrices, vectors with partially known support and many others. We will present some of these applications in a few forthcoming papers.

\begin{singlespace}
\bibliographystyle{plain}
\bibliography{LqThrBndsRefs}
\end{singlespace}

\end{document}